\documentclass{article}

\usepackage[english]{babel}

\usepackage[letterpaper,top=2cm,bottom=2cm,left=3cm,right=3cm,marginparwidth=1.75cm]{geometry}

\usepackage{graphicx} 
\usepackage{xcolor}
\usepackage{algorithm}
\usepackage{algorithmic}
\usepackage{amssymb} 
\usepackage{amsthm}
\usepackage{amsmath} 
\usepackage{authblk}
\usepackage{parskip}
\usepackage{hyperref}
\usepackage{url}

\newcommand{\cmin}{c_{\text{min}}}
\DeclareMathOperator{\poly}{poly}

\usepackage{tikz}
\usetikzlibrary{calc}
\usetikzlibrary{patterns}
\usetikzlibrary{math}

\newtheorem{lemma}{Lemma}
\newtheorem{corollary}[lemma]{Corollary}
\newtheorem{theorem}[lemma]{Theorem}

\usetikzlibrary{decorations.pathreplacing} 
\tikzset{vertexA/.style={circle,fill=black,minimum size=5pt,inner sep=0pt, font=\tiny}}
\tikzstyle{vertexB}=[circle,fill=black,minimum size=5pt,inner sep=0pt, font=\tiny]

\title{On the Assadi–Liu–Tarjan Auction Algorithm for Bipartite Matching: Simplification, Alternative Analysis, and Hard Instance} 

\date{}

\author{Christian Konrad\thanks{\texttt{christian.konrad@bristol.ac.uk}, School of Computer Science, University of Bristol, UK. }\hspace{20px}Kheeran K. Naidu\thanks{\texttt{kheeran.naidu@gmail.com}, Qworky Research, London, UK.}\hspace{20px}Archie Walton\thanks{\texttt{archie.walton@bristol.ac.uk}, School of Computer Science, University of Bristol, UK. Supported by an EPSRC DTP studentship.}\hspace{20px}Yiquan Wang\thanks{\texttt{ericwangyiquan@gmail.com}, University of California, San Diego, US.}}

\begin{document}

\maketitle

\begin{abstract}
Assadi, Liu, and Tarjan [SOSA'21] gave an auction algorithm that outputs a $(1-\epsilon)$-approximation to \textsf{Maximum Matching} in bipartite graphs. Their algorithm computes a sequence of $O(\frac{1}{\epsilon^2})$ maximal matchings in subgraphs of the input graph and can be implemented in the multi-pass streaming setting with $O(\frac{1}{\epsilon^2})$ passes in a straightforward manner, which constitutes the state-of-the-art pass/approximation trade-off result in the multi-pass streaming setting. Their analysis uses tools from combinatorial auctions and, at its heart, relies on a clever potential function argument. Their proof, however, provides only limited insight into the inner workings of the algorithm.  

\vspace{0.1cm}

In this paper, we revisit the \textsc{ALT}-algorithm and present the following contributions:

\vspace{0.1cm}

\begin{enumerate}
    \item \textbf{Simplification.} The \textsc{ALT}-algorithm is built upon a {\em freezing mechanism} where vertices on one side of the bipartition that have already been rematched $\Theta(\frac{1}{\epsilon})$ times over the course of the algorithm remain matched to their current partner forever. We show that this mechanism is in fact unnecessary, i.e., no special treatment of such vertices is needed. With the freezing mechanism removed, the parameter $\epsilon$ now solely determines the total number of iterations/maximal matching computations, which provides the option of adaptively refining $\epsilon$ as the algorithm runs. 
    \item \textbf{Alternative Analysis.} We give an alternative analysis of the algorithm that is based on augmenting paths. Beyond the auction-perspective of the algorithm as established by Assadi et al., our analysis allows for a reinterpretation as one that follows the  traditional approach of searching for and eliminating augmenting paths. Our analysis also copes with the removal of the freezing mechanism in a natural way, whereas the analysis of Assadi et al. strictly depends on its use. 
    \item \textbf{Hard Instance.} We provide the first hard instance on which the algorithm requires $\Omega(\frac{1}{\epsilon^2})$ iterations/maximal matching computations. The instance is a simple path graph, where we exhibit a cyclic behaviour that prevents fast progress. Hard instances for this algorithm therefore do not necessarily have to be dense.     
\end{enumerate}
\end{abstract}

\section{Introduction} 
A {\em matching} in a graph  is a subset of vertex-disjoint edges. A matching is {\em maximal} if it cannot be enlarged by adding an edge outside of the matching to it, and a matching is {\em maximum} if it is of largest possible size. Then, the \textsf{Maximum Matching} problem (\textsf{MM}) consists of computing a maximum matching, and, for any $0 < \alpha \le 1$, an $\alpha$-approximation algorithm to \textsf{MM} is an algorithm that outputs a matching of size at least $\alpha$ times the size of a maximum matching.

Motivated by the fact that, in various restricted computational models, computing a maximal matching is an easy task, there has been growing interest in algorithms for \textsf{MM} approximation that are solely based on computing a sequence of maximal matchings in subgraphs of the input graph so that the union of these maximal matchings contains a non-trivial approximation to \textsf{MM} \cite{kk20,kns23}. For example, in the {\em semi-streaming model} of computation \cite{fkmsz05}, where an algorithm performs multiple passes over the edges of an $n$-vertex input graph while maintaining a memory of size $O(n \poly \log n)$, a single pass is sufficient for computing a maximal matching. In the $O(n)$-memory {\em massively parallel computation model} (MPC model) \cite{ksv10}, where the input graph is partitioned between multiple machines with memory $O(n)$ each and the machines exchange messages in synchronous rounds, a maximal matching can be computed in $O(\log \log n)$ rounds \cite{bhhz19}. Hence, any  algorithm that finds an approximation to \textsf{MM} solely by computing $p$ maximal matchings in subgraphs of the input graph can therefore give algorithms in restricted computational models with only a factor $p$ overhead over the complexity of computing a maximal matching\footnote{As long as the subgraphs on which the maximal matching algorithm is to be executed can be specified/computed in the respective computational model.}.

Assadi, Liu, and Tarjan \cite{alt21} gave such an algorithm for bipartite graphs (henceforth the \textsc{ALT}-algorithm). For any small $\epsilon > 0$, the \textsc{ALT}-algorithm produces a $(1-\epsilon)$-approximation to \textsf{MM} and is solely based on the sequential computation of $O(\frac{1}{\epsilon^2})$ maximal matchings. When implemented in the semi-streaming model, the algorithm uses $O(\frac{1}{\epsilon^2})$ passes and constitutes the state-of-the-art bipartite matching algorithm in this model. We, however, note that streaming algorithms with pass complexity $O(\frac{1}{\epsilon} \cdot \log n)$ are also known \cite{ag18,ajjst22,a24}, which run faster when $\epsilon = o(\frac{1}{\log n})$.

\vspace{0.1cm} 

The \textsc{ALT}-algorithm as presented in \cite{alt21} is phrased using terminology from combinatorial auctions, and their analysis also uses ideas from this area. We observe that this is very different to previous matching algorithms in the streaming model, which are mostly based on finding and eliminating augmenting paths \cite{kns23,ab21,kn21,k18, ekms12,kmm12,m05},  on linear programming \cite{a24,ag13}, and, more recently, on the edge-degree constrained subgraph technique \cite{b24,ab21}. Regarding $(1-\epsilon)$-approximation multi-pass algorithms for bipartite graphs, Eggert et al. \cite{ekms12} previously gave an augmenting paths-based algorithm that runs in $O(\frac{1}{\epsilon^5})$ passes, and Ahn and Guha \cite{ag13} gave a linear programming based algorithm that runs in $O(\frac{1}{\epsilon^2} \log \log(\frac{1}{\epsilon}))$ passes. 

\vspace{0.1cm}

\noindent \textbf{Our Results.} 
The purpose of this paper is to provide a deeper understanding of the \textsc{ALT}-algorithm. We present the following contributions:

\begin{enumerate}
 \item \textbf{Simplification.} The \textsc{ALT}-algorithm operates on a bipartite graph $G=(A, B, E)$ and maintains an initially empty matching $M$ that, at the end of the algorithm, constitutes a $(1-\epsilon)$-approximation. The algorithm is round-based, and, in each round of the algorithm, a matching $M'$ between the currently unmatched $A$-vertices and their neighbors is computed in a carefully chosen {\em demand subgraph} (see Section~\ref{sec:algorithm} for precise definitions). The matching $M$ is then updated  by setting 
 $$M \gets M' \cup \{e \in M \ : \ e \text{ not incident to an edge in $M'$ }  \} \ .$$ 
 One property of proceeding in this way is that, over the course of the algorithm, $B$-vertices change their partners many times. The \textsc{ALT}-algorithm is built upon a {\em freezing mechanism} that effectively freezes a $B$-vertex to its current match once it has been rematched $\Theta(\frac{1}{\epsilon})$ times. This somewhat artificial rule is crucial for Assadi et al.'s analysis to work. 
 
 We show that the freezing mechanism is in fact unnecessary and no special treatment of vertices that have changed their partners often is required. Removing the freezing mechanism simplifies the algorithm and has the added benefit that the parameter $\epsilon$ now solely determines the number of iterations/maximal matching computations of the algorithm. This implies that adaptive refinements of $\epsilon$ while the algorithm runs are now possible.

 \item \textbf{Alternative Analysis.} While our analysis reuses some of the facts established by Assadi et al., we employ an entirely different argument at its core that deals with the removal of the freezing mechanism in a natural way. We establish a connection between the number of times $c(b)$ that a $B$-vertex $b$ has changed partners and the length of any augmenting path that runs through $b$. More precisely, if there exists an augmenting path starting at an unmatched vertex $a \in A$ that goes through the vertex $b$ then this path is of length at least $2c(b) + 1$. Hence, the lengths of augmenting paths increase over the course of the algorithm whenever a $B$-vertex is rematched. Since longer augmenting paths imply larger matchings, this connection immediately explains how the algorithm makes progress.

 \item \textbf{Hard Instance.} We give the first hard instance on which the \textsc{ALT}-algorithm requires $\Omega(\frac{1}{\epsilon^2})$ iterations. Our instance is the path graph, where we demonstrate a cyclic behaviour that prevents fast progress. In this instance,  the algorithm repeatedly cycles through very similar matchings many times. Interestingly, most edges of the path are inserted and removed from the current matching $\Theta(\frac{1}{\epsilon})$ times, a key property that demonstrates  that hard instances for this algorithm do not necessarily have to be dense. 
 

\end{enumerate}

\vspace{0.1cm}
\noindent \textbf{Outline.} In Section~\ref{sec:algorithm}, we discuss the \textsc{ALT}-algorithm, and we give an augmenting-paths based analysis of the algorithm in Section~\ref{sec:analysis}. Then, in Section~\ref{sec:lb}, we show that the algorithm may require $\Omega(\frac{1}{\epsilon^2})$ passes on a path graph. Finally, we conclude in Section~\ref{sec:conclusion}.


\section{Algorithm}\label{sec:algorithm}
The \textsc{ALT}-algorithm is depicted in Algorithm~\ref{alg:1-eps}. 

\begin{algorithm}
	\caption{Assadi-Liu-Tarjan Algorithm: ($1-\epsilon$)-approximation for \textsf{MM} on bipartite graphs} \label{alg:1-eps}
	\textbf{Input:} Bipartite input graph $G = (A,B,E)$, parameter $\epsilon > 0$

	\begin{algorithmic}[1]        
		\STATE Let $M = \varnothing$ and, for all $b \in B$, let $c(b) = 0$

        \FOR{\textbf{each} round $r = 1, 2, \ldots, 49/\epsilon^2$}
            \STATE For every $a \in A \setminus A(M)$, let $\cmin(a):= \text{min}\{ c(b) : b \in \Gamma(a) \}$            
            \STATE \label{line:four} Let $H_r \subseteq G$ be the subgraph spanned by all edges $(a,b) \in (A \setminus A(M)) \times B$ such that $c(b) = \cmin(a)$ 
            \STATE Compute a maximal matching $N_r$ in $H_r$
            \STATE $M \gets \{ ab \in M \ : \ b \notin B(N_r) \}  \cup N_r$
            \STATE Increment the count $c(b)$ of each $b \in B(N_r)$ by $1$
        \ENDFOR
		\RETURN $M$
	\end{algorithmic}
\end{algorithm}

The listing differs in two aspects from its original version in \cite{alt21}. First, in our version, the quantity $c(b)$, for every $b \in B$, when increased is increased by one, while this quantity is increased by $\epsilon$ in \cite{alt21}. The reason for scaling these values is that our interpretation of $c(b)$ is different as in \cite{alt21} -- see further below for a discussion.

Second, and more importantly, the original version in \cite{alt21} {\em freezes} $B$-vertices once they have been rematched $\Theta(\frac{1}{\epsilon})$ times. This translates into adding the constraint "and $c(b) \le \frac{3}{\epsilon}$" to the end of Line~\ref{line:four}. As previously discussed, one of our key contributions is the insight that the freezing mechanism is unnecessary and can be removed.

We note that removing the freezing mechanism has a significant benefit. In our version of the \textsc{ALT}-algorithm, the parameter $\epsilon$ solely determines the number of iterations, while, in its original version by Assadi et al., the parameter $\epsilon$ also determines the moment at which vertices are frozen. Hence, in our version, the execution of the algorithm for a parameter $\epsilon' > \epsilon$ is a prefix of the execution of the algorithm with parameter $\epsilon$, which means that the precision parameter $\epsilon$ can be updated while the algorithm is running. This is not possible in its original version.

The algorithm works as follows. Given the bipartite input graph $G=(A, B, E)$, the algorithm maintains a matching $M$. The matching $M$ evolves over the course of the algorithm but maintains the invariant that if a vertex $b \in B$ is matched in $M$ at some moment then the vertex $b$ remains matched throughout the algorithm. We note that $b$ may change mates 
along the way, however, it will never become unmatched again. The matching $M$ thus cannot decrease in size. 

The algorithm maintains, for each $b \in B$, a counter $c(b)$, which counts the number of times the vertex has changes its mate in $M$. In its original version in \cite{alt21}, once $c(b)$ reaches the count $\Theta(\frac{1}{\epsilon})$ then $b$ ceases to swap partners and remains matched to the current partner until the end of the algorithm. In our version, this step is unnecessary and $b$ continues to swap partners until the algorithm terminates.

In each iteration of the algorithm, the currently unmatched $A$-vertices compete for a new mate, which often requires unmatching a currently matched $A$-vertex. We define the \textit{minimum neighbourhood cost} $\cmin(a) := \min \{c(b) \ : \ b \in \Gamma(a) \}$ for each vertex $a \in A$, which is the smallest number of times that any neighbour of $a$ has changed mates so far. In each round $r$, each unmatched $A$-vertex $a \in A \setminus A(M)$ first computes $\cmin(a)$. Then, a {\em demand subgraph} $H_r$ is established, where, for each unmatched $a \in A \setminus A(M)$, every edge $ab$ is included such that $c(b) = \cmin(a)$. Next, a maximal matching $N_r$ in $H_r$ is computed and the current matching $M$ is updated so that every edge of $M$ that is incident to a matched $B$-vertex in $N_r$ is replaced with the corresponding edge in $N_r$. 
Finally, the count $c(b)$ of every vertex $b$ that has received a new mate is increased by $1$ to reflect the previous step.

Phrased in the language of auctions as in \cite{alt21}, the $A$-vertices constitute the set of bidders and the $B$-vertices constitute the set of items. Then, for an item $b \in B$, the quantity $c(b)$ is the current price of an item, which increases by $1$ every time the item is rematched to a different bidder. For a bidder $a \in A$, the quantity $\cmin(a)$ corresponds to the price of the cheapest item in the neighbourhood of $a$, i.e., items that the bidder $a$ could acquire. The demand subgraph $H_r$ thus only consists of edges between unmatched bidders and the cheapest items available to them.

\vspace{0.1cm}
\noindent \textbf{Basic Properties.}
We highlight some basic facts about the algorithm that are easy to verify and that were also used by Assadi et al.\ \cite{alt21}:
\begin{enumerate}
 \item[\textbf{P1:}] Once a $B$-vertex is matched, it will never become unmatched again. The matching $M$ therefore never decreases in size.
 \item[\textbf{P2:}] For any vertices $a \in A$ and $b \in B$, the quantities $\cmin(a)$ and $c(b)$, respectively, never decrease.
 \item[\textbf{P3:}] Consider an unmatched $A$-vertex $a \in A \setminus A(M)$ at the start of any iteration, and suppose that it remains unmatched by the end of the iteration. Then, the value $\cmin(a)$ increases by one in this iteration.
\end{enumerate}
Furthermore, we say that an unmatched $A$-vertex $a$ is {\em saturated} if $\cmin(a) \ge 6 / \epsilon$. An $A$-vertex that is not saturated is {\em unsaturated}. 

\section{Analysis} \label{sec:analysis}
Let $OPT$ be an arbitrary but fixed maximum matching in the input graph $G=(A, B, E)$. We write $A_{\text{OPT}} := A(OPT)$ to denote the $A$-vertices incident to $OPT$. 

Assadi et al.'s analysis has two main components.  First, they argue that if a specific configuration of the variables is achieved then the matching $M$ must be of size at least $(1-\epsilon)|OPT|$ ({\em approximation part}). Second, they argue that such a configuration must necessarily be reached at some moment during the execution ({\em progress part}).

In our analysis below, we introduce augmenting paths-based arguments. These arguments are used to both give an entirely different proof of the approximation part, and to augment Assadi et al.'s arguments for the progress part in order to show that the algorithm equally works without the freezing mechanism.

Assadi et al.'s key lemma regarding the approximation factor states the following. They define the notion of {\em $\epsilon$-happiness} of a bidder, which translates to saying that an $A$-vertex is $\epsilon$-happy either if it is matched or if its $\cmin(a)$ value has reached $\Theta(1/\epsilon)$ -- the latter of which corresponds to our notion of an unmatched $A$-vertex being saturated. The interpretation of this notion is that if a bidder is  matched then it is clearly happy, while if it is not matched but all items in their neighbourhood have become quite expensive then the unmatched $A$-vertex is still {\em relatively} happy, or $\epsilon$-happy, as they would need to spend a large amount of money to purchase the item. They prove that if in some iteration of the algorithm there are $(1-\frac{1}{2} \epsilon)|OPT|$ $\epsilon$-happy bidders then the resulting matching is a $(1-\epsilon)$-approximation. Their proof uses a potential function argument, summing the ``utilities'' of the $\epsilon$-happy bidders and giving both upper and lower bounds on this quantity.

Our key insight into an augmenting paths-based analysis of the approximation factor is as follows. Consider an unmatched vertex $a \in A_{\text{OPT}} \setminus A(M)$ with its min-cost $\cmin(a)$. We observe that the quantity $\cmin(a)$ directly translates into a lower bound on the length of a potential augmenting path in $M \oplus OPT = (M \setminus OPT )\cup (OPT \setminus M)$ that starts at vertex $a$, i.e., we prove that if an augmenting path starts at $a$ then it is of length at least $2 \cmin(a) + 1$. Hence, if an unmatched vertex has a large $\cmin$ value then the damage incurred by the fact that $a$ is not matched is small since the potential augmenting path starting at $a$ is long. Recall that, in an augmenting path of length $2\ell + 1$, we have only one unmatched $A$-vertex and $\ell$ matched $A$-vertices. Thus, locally, the approximation factor on this path is $\ell / (\ell + 1)$. Hence, if we can ensure that $\cmin(a) = \Omega(\frac{1}{\epsilon})$ holds for most unmatched $A_{\text{OPT}}$-vertices  then the current matching $M$ is a $(1-\epsilon)$-approximation. We will see that this follows from the progress part of the analysis.

Regarding the progress part of the analysis, we use a similar argument to Assadi et al. However, we include the insight that the number of unmatched $A_{\text{OPT}}$-vertices with large min-cost value which begin augmenting paths in $M\oplus OPT$ is bounded, which is a property that follows from our augmenting paths-based analysis. In more detail, we show that there are at most $|OPT|/k$ unmatched $A_{\text{OPT}}$-vertices with min-cost value $\Omega(k)$ at the beginning of augmenting paths, for integers $k \ge 1$. This property is crucial to show that the freezing mechanism is not needed. 

We will now give the proof of our key lemma. Throughout the following proofs, we will define $M_r$ to be the matching at the beginning of iteration $r$, and
\begin{align*}
    S_r := \{a \in A_{\text{OPT}} : a \text{ is the starting vertex of an augmenting path in }M_r\oplus OPT \}.
\end{align*}

\vspace{0.2cm} 

\begin{lemma}\label{lem:aug-path-length}
 Consider the state of the algorithm at the beginning of any iteration $r$. For any $k \ge 0$, let $\mathcal{P}$ be an augmenting path of length $2k+1$ in $M_r \oplus OPT$, and let $a_0 \in S_r$ be the unmatched $A$-vertex in $\mathcal{P}$. Then:
 $$k \ge \cmin(a_0) \ . $$
 Furthermore, the set $M_r \oplus OPT$ contains at most $\frac{1}{k+1} \cdot |OPT|$  augmenting paths that start at $S_r$-vertices $a$ with $\cmin(a) \ge k$.
\end{lemma}
\begin{proof}
Let $\mathcal{P} = a_0 b_1 a_1 b_2 a_2 \dots b_k a_k b_{k+1} \in M_r \oplus OPT$ be an augmenting path of length $2k + 1$, i.e., we have $a_ib_{i+1} \in OPT$ and $a_i b_i \in M_r$, for every $i \in [k]$. We first observe that, by definition of $\cmin(a_i)$ and the fact that $b_{i+1}$ is in the neighbourhood of $a_i$, it holds that
 $c(b_{i+1}) \ge \cmin(a_i)$. 

Next, we will argue that, for every $0 \le i \le k-1$, we have 
    $\cmin(a_i)-1 \le \cmin(a_{i+1}).$
Indeed, as argued above, we have $c(b_{i+1}) \ge \cmin(a_i)$. Next, observe that $a_{i+1}$ is matched in $M$ to $b_{i+1}$. Consider the iteration $s$ when $a_{i+1}$ was most recently matched to $b_{i+1}$.  Then, denote by $c^s(b_{i+1})$ the value of $c(b_{i+1})$ in this iteration, and, similarly, denote by $\cmin^s(a_i)$ the value of $\cmin(a_i)$ in this iteration. In iteration $s$, the value of $c^s(b_{i+1})$ was increased to its current value $c(b_{i+1})$, i.e., we have $c^s(b_{i+1}) = c(b_{i+1}) - 1$. Since in iteration $s$, the edge $(a_{i+1}, b_{i+1})$ was contained in $H_s$, we have that $\cmin^s(a_{i+1}) = c^s(b_{i+1}) (=  c(b_{i+1}) - 1)$. Last, since $\cmin(a_{i+1}) \ge \cmin^r(a_{i+1})$ (Property~\textbf{P2}), we obtain the claimed bound.

The previous bound implies that $\cmin(a_k) \ge \cmin(a_0) - k$, and this in turn implies that $c(b_{k+1}) \ge \cmin(a_0) - k$. Since $b_{k+1}$ is unmatched, which, by Property~\textbf{P1}, means that $b_{k+1}$ has not previously been matched, we have that $c(b_{k+1}) = 0$. We thus obtain $0 \ge \cmin(a_0) - k$ 
as desired.

Next, consider the set of augmenting paths $\mathcal{P}_{\ge k} \subseteq M_r \ \oplus \ OPT$ starting at a vertex $a \in S_r$ with $\cmin(a) \ge k$. Observe that these paths are vertex-disjoint, and, as discussed above, each of these paths visits at least $k+1$ $A_{\text{OPT}}$ vertices. Hence, there can be at most $|A_{\text{OPT}}| / (k+1) = |OPT| / (k+1)$ such paths, which completes the proof.
\end{proof}

 

Using the insight from the previous lemma, we show that it is enough to find an upper bound on the number of unmatched $A_\text{OPT}$ vertices that start augmenting paths and have low neighbourhood cost in order to establish the approximation factor of the algorithm.

\begin{corollary}\label{cor:approx}
 Consider before some iteration $r$ where the number of vertices of $S_r$ with minimum neighbourhood cost less than $2/\epsilon$ is at most $\epsilon / 2 \cdot |OPT|$. Then, the current matching $M_r$ is a $(1-\epsilon)$-approximation.
\end{corollary}

\begin{proof}
 Consider the set of augmenting paths $\mathcal{A}$ in $M_r \oplus OPT$, each starting at a vertex in $S_r$. Partition $S_r$ between two sets, $S_r^+ \subseteq S_r$ being the vertices in $S_r$ with minimum neighbourhood cost at least $2/\epsilon$, and $S_r^- = S_r \setminus S_r^+$. By Lemma~\ref{lem:aug-path-length}, we have $|S_r^+| \le \frac{1}{\frac{2}{\epsilon}+1}|OPT| \le \frac{\epsilon}{2} |OPT|$. Then: 
 $$|\mathcal{A}| = |S_r| = |S_r^+| + |S_r^-| \le \epsilon / 2 \cdot |OPT| + \epsilon / 2 \cdot |OPT| = \epsilon \cdot |OPT| \ . $$
 Since $|M| + |\mathcal{A}| = |OPT|$, the result follows.
\end{proof}

\begin{theorem}
 Our version of the \textsc{ALT}-algorithm listed in Algorithm~\ref{alg:1-eps} is a $(1-\epsilon)$-approximation algorithm for \textsf{MM} on bipartite graphs.
\end{theorem}
\begin{proof}
Let $M_{\text{out}}$ be the final matching output by the algorithm, and let $B_{\text{out}} = B(M_\text{out})$. 
In this proof, we denote by $c^r(b)$ the quantity $c(b)$ in the beginning of iteration $r$. 


We will first argue that the number of iterations $X$ in which a large number of vertices in $S_r$ are matched is bounded above. 
Let $\hat{N_r}$ be the set of edges in $N_r$, the matching computed in round $r$, which contain a vertex of $S_r$. We will show that the number of iterations in which $|\hat{N_r}| \geq \frac{\epsilon}{3}|OPT|$ is at most ${36}/{\epsilon^2}$.

To this end, consider such an iteration $r$ and partition the $B(\hat{N_r})$-vertices into $B_{<}(\hat{N_r})$ and $B_{\ge}(\hat{N_r})$ such that each $b \in B_{<}(\hat{N_r})$ has $c^r(b) < 6/\epsilon$ and each $b \in B_{\ge}(\hat{N_r})$ has $c^r(b) \ge 6/\epsilon$. Observe that, by Lemma~\ref{lem:aug-path-length}, there are at most $(\epsilon / 6) \cdot |OPT|$ vertices of $S_r$ with minimum neighbourhood cost at least $\frac{6}{\epsilon}$. Hence, at most $(\epsilon / 6) \cdot |OPT|$ vertices in $B_{\ge}(\hat{N_r})$ can be matched in iteration $r$ and increase their costs, since only unmatched saturated $A$-vertices can match to these vertices. This in turn implies that at least $(\epsilon / 6) \cdot |OPT|$ vertices of $B_{<}(\hat{N_r})$ are matched in the current iteration. 

Observe further that $B(\hat{N_r}) \subseteq B_{\text{out}}$, which follows from Property \textbf{P1}. Over the course of the $X$ iterations, we have thus seen at least 
\begin{align}
X \cdot \frac{\epsilon}{6} \cdot |OPT| \label{eqn:391}    
\end{align} 
increases of the costs of $B_{\text{out}}$-vertices $b$ when they had a current value $c(b) < \frac{6}{\epsilon}$. This number, however, is also trivially bounded from above by $\frac{6}{\epsilon} \cdot |B_{\text{out}}| \le \frac{6}{\epsilon} \cdot |OPT|$, which together with the lower bound in \ref{eqn:391} implies that $X \le \frac{36}{\epsilon^2}$ as desired.

The previous claim implies that the number of iterations where $|\hat{N_r}| < \frac{\epsilon}{3} |OPT|$ is at least $\frac{49}{\epsilon^2} - \frac{36}{\epsilon^2} = \frac{13}{\epsilon^2}$ (recall that Algorithm~\ref{alg:1-eps} overall executes $\frac{49}{\epsilon^2}$ iterations). 
Suppose for the sake of a contradiction that $M_{\text{out}}$ is not a $(1-\epsilon)$-approximation. Then, by Corollary~\ref{cor:approx}, in each iteration, there are at least $\frac{\epsilon}{2} \cdot |OPT|$ unmatched vertices $a \in S_r$  with $\cmin(a) < 2 / \epsilon$. Since the number of $S_r$ vertices matched in any of these $13/\epsilon^2$ iterations is less than $\epsilon/3 \cdot |OPT|$, in each such iteration there are at least $\frac{\epsilon}{2} \cdot |OPT| - \frac{\epsilon}{3} \cdot |OPT| = \frac{\epsilon}{6}  \cdot |OPT|$ vertices of $S_r$ with min-cost smaller than  $2 / \epsilon$ that are not matched. Observe that, by Property~\textbf{P3}, each such vertex increases its $\cmin$ value. Hence, over the course of the algorithm, we observe at least 
\begin{align}
13 / \epsilon^2 \cdot \frac{\epsilon}{6} \cdot |OPT| = \frac{13}{6 \epsilon} |OPT| \label{eqn:555}
\end{align} 
increases in the min-cost value of $S_r$ vertices when their current min-cost was below $2/\epsilon$. Since $S_r$ is always a subset of $A_\text{OPT}$, the total cost increase on vertices with minimum neighbourhood cost at most $2/\epsilon$ is bounded above by $\frac{2}{\epsilon} \cdot |A_\text{OPT}| = \frac{2}{\epsilon} \cdot |OPT|$. This yields a contradiction to the lower bound established in Equation~\ref{eqn:555}. Hence, $M_{\text{out}}$ is a $(1-\epsilon)$-approximation.
\end{proof}

\section{Hard Instance} \label{sec:lb}

We now show that, even on path graphs, the \textsc{ALT}-algorithm may require $\Omega(n^2)$ iterations to obtain a maximum matching. This directly translates to a proof that the number of iterations required for a $(1-\epsilon)$-approximation to \textsf{MM} is $\Theta(\frac{1}{\epsilon^2})$.

We use the family of path graphs on $6x-2$ vertices for any $x \geq 2$, with $3x-1$ $A$ and $B$ vertices, respectively. We label the $B$-vertex that is an endpoint of the path as $b_0$. We label the next $B$-vertex along the path as $b_1$, and label the rest of the $B$-vertices by $b_2,...,b_{3x-2}$ in the same manner. We then label the $A$-vertices such that $a_i$ is adjacent to $b_i$ and $b_{i+1}$. The vertex $a_{3x-2}$ is the $A$-vertex that is only adjacent to $b_{3x-2}$. This is the other endpoint of the path. We split the $A$-vertices into three sets $A_0, A_1$ and $A_2$ such that $A_k := \{a_\ell : \ell \equiv k \mod{3}\}$. We split the $B$-vertices into $B_0, B_1, B_2$ in the same way.

We note that the only maximum matching in this graph is the matching between $a_k$ and $b_k$ for each $k$ from $0$ to $3x-2$. We set $\epsilon = \frac{1}{3x}$ so that only the singular maximum matching is a $(1-\epsilon)$-approximate matching. 

    We now describe the four `setup' rounds. An example of a path on 16 vertices after each of these rounds is given in Figure~\ref{fig:setup}. 

    \begin{enumerate}
        \item[\textbf{Round 1:}] Let the matching returned be the set of edges between $A_0$ and $B_1$, as well as the edges between $A_2$ and $B_2$. We note that this matching is maximal. The vertices in $A_1$ and $B_0$ remain unmatched. 
        \item[\textbf{Round 2:}] Each vertex in $A_1$ has two neighbours of cost 1. When a vertex has multiple min-cost neighbours, we will take the $B$-vertex of larger index, represented by the diagonal edges in Figure~\ref{fig:setup}. We add all of the edges between $A_1$ and $B_2$. Each vertex in $B_2$ has its cost increased to 2, and the vertices in $A_2$ become unmatched. Additionally, the edge between $a_{3x-2}$ and $b_{3x-2}$ is matched.
        \item[\textbf{Round 3:}]Each of the vertices in $A_2$ is adjacent to an unmatched vertex in $B_0$, each of which are added to the matching. The only $A$-vertex that is not matched after this round is $a_{3x-3}$, as it is blocked from matching to $b_{3x-3}$ by $a_{3x-4}$.
        \item[\textbf{Round 4:}]$a_{3x-3}$ must match to $b_{3x-3}$, $a_{3x-4}$ becomes unmatched.
    \end{enumerate}

    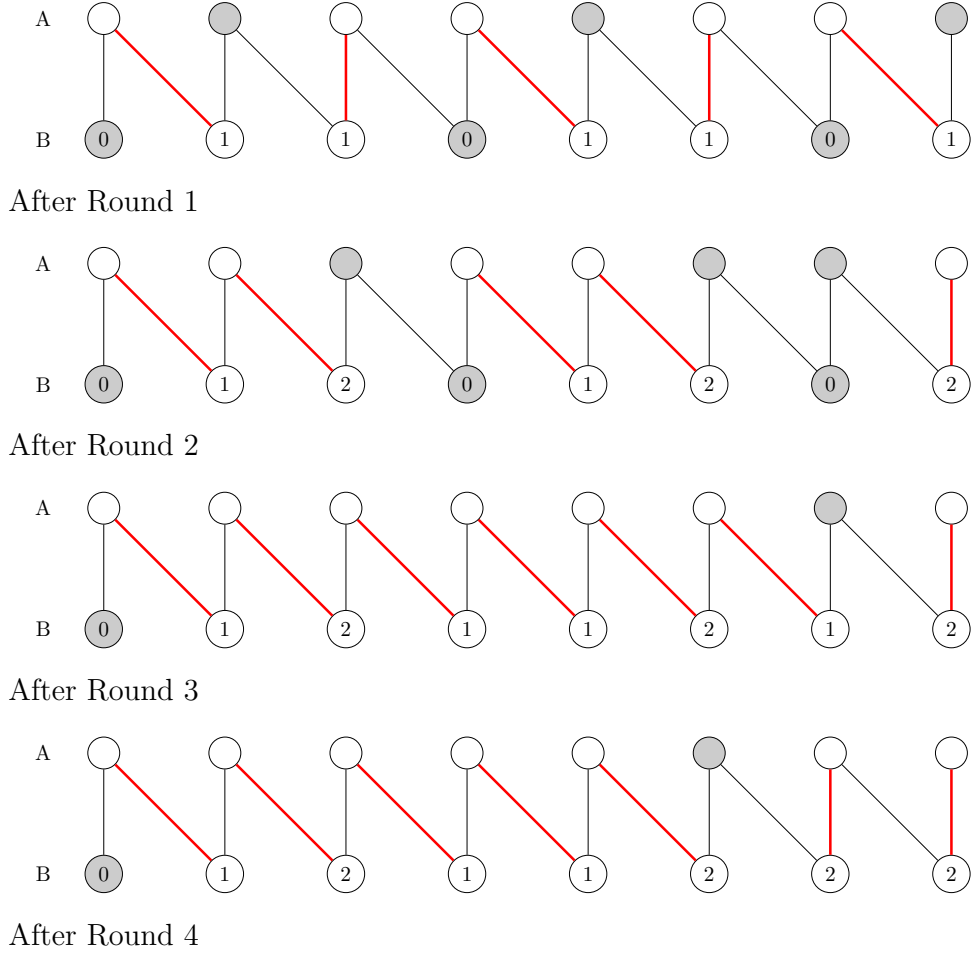
\begin{figure}[h]
        \centering
        
        \scalebox{.8}{
            \begin{tikzpicture}
    
                \node[circle, draw, minimum size=1.5em] (A1) at (0,0) {};
                \node[circle, draw, minimum size=1.5em, fill=black!20] (A2) at (0,-2) {0};
                \node[circle, draw, minimum size=1.5em, fill=black!20] (B1) at (2,0) {};
                \node[circle, draw, minimum size=1.5em] (B2) at (2,-2) {1};
                \node[circle, draw, minimum size=1.5em] (C1) at (4,0) {};
                \node[circle, draw, minimum size=1.5em] (C2) at (4,-2) {1};
                \node[circle, draw, minimum size=1.5em] (D1) at (6,0) {};
                \node[circle, draw, minimum size=1.5em, fill=black!20] (D2) at (6,-2) {0};
                \node[circle, draw, minimum size=1.5em, fill=black!20] (E1) at (8,0) {};
                \node[circle, draw, minimum size=1.5em] (E2) at (8,-2) {1};
                \node[circle, draw, minimum size=1.5em] (F1) at (10,0) {};
                \node[circle, draw, minimum size=1.5em] (F2) at (10,-2) {1};
                \node[circle, draw, minimum size=1.5em] (G1) at (12,0) {};
                \node[circle, draw, minimum size=1.5em, fill=black!20] (G2) at (12,-2) {0};
                \node[circle, draw, minimum size=1.5em, fill=black!20] (H1) at (14,0) {};
                \node[circle, draw, minimum size=1.5em] (H2) at (14,-2) {1};
    
                \node[] (A) at (-1,0) {A};
                \node[] (B) at (-1,-2) {B};
                \node[] (R1) at (0,-3) {\Large After Round 1};
            
                \draw[] (A1) -- (A2);
                \draw[] (B1) -- (B2);
                \draw[color=red, very thick] (A1) -- (B2);
                \draw[] (B1) -- (C2);
                \draw[color=red, very thick] (C1) -- (C2);
                \draw[] (C1) -- (D2);
                \draw[] (D1) -- (D2);
                \draw[color=red, very thick] (D1) -- (E2);
                \draw[] (E1) -- (E2);
                \draw[] (E1) -- (F2);
                \draw[color=red, very thick] (F1) -- (F2);
                \draw[] (F1) -- (G2);
                \draw[] (G1) -- (G2);
                \draw[color=red, very thick] (G1) -- (H2);
                \draw[] (H1) -- (H2);
            
            \end{tikzpicture}
        }
        
        \vspace{10px}

        \scalebox{.8}{
            \begin{tikzpicture}

            \node[circle, draw, minimum size=1.5em] (A1) at (0,0) {};
            \node[circle, draw, minimum size=1.5em, fill=black!20] (A2) at (0,-2) {0};
            \node[circle, draw, minimum size=1.5em] (B1) at (2,0) {};
            \node[circle, draw, minimum size=1.5em] (B2) at (2,-2) {1};
            \node[circle, draw, minimum size=1.5em, fill=black!20] (C1) at (4,0) {};
            \node[circle, draw, minimum size=1.5em] (C2) at (4,-2) {2};
            \node[circle, draw, minimum size=1.5em] (D1) at (6,0) {};
            \node[circle, draw, minimum size=1.5em, fill=black!20] (D2) at (6,-2) {0};
            \node[circle, draw, minimum size=1.5em] (E1) at (8,0) {};
            \node[circle, draw, minimum size=1.5em] (E2) at (8,-2) {1};
            \node[circle, draw, minimum size=1.5em, fill=black!20] (F1) at (10,0) {};
            \node[circle, draw, minimum size=1.5em] (F2) at (10,-2) {2};
            \node[circle, draw, minimum size=1.5em, fill=black!20] (G1) at (12,0) {};
            \node[circle, draw, minimum size=1.5em, fill=black!20] (G2) at (12,-2) {0};
            \node[circle, draw, minimum size=1.5em] (H1) at (14,0) {};
            \node[circle, draw, minimum size=1.5em] (H2) at (14,-2) {2};

            \node[] (A) at (-1,0) {A};
            \node[] (B) at (-1,-2) {B};
            \node[] (R1) at (0,-3) {\Large After Round 2};
        
            \draw[] (A1) -- (A2);
            \draw[] (B1) -- (B2);
            \draw[color=red, very thick] (A1) -- (B2);
            \draw[color=red, very thick] (B1) -- (C2);
            \draw[] (C1) -- (C2);
            \draw[] (C1) -- (D2);
            \draw[] (D1) -- (D2);
            \draw[color=red, very thick] (D1) -- (E2);
            \draw[] (E1) -- (E2);
            \draw[color=red, very thick] (E1) -- (F2);
            \draw[] (F1) -- (F2);
            \draw[] (F1) -- (G2);
            \draw[] (G1) -- (G2);
            \draw[] (G1) -- (H2);
            \draw[color=red, very thick] (H1) -- (H2);

            \end{tikzpicture}
        }
        
        \vspace{10px}

        \scalebox{.8}{
            \begin{tikzpicture}
    
                \node[circle, draw, minimum size=1.5em] (A1) at (0,0) {};
                \node[circle, draw, minimum size=1.5em, fill=black!20] (A2) at (0,-2) {0};
                \node[circle, draw, minimum size=1.5em] (B1) at (2,0) {};
                \node[circle, draw, minimum size=1.5em] (B2) at (2,-2) {1};
                \node[circle, draw, minimum size=1.5em] (C1) at (4,0) {};
                \node[circle, draw, minimum size=1.5em] (C2) at (4,-2) {2};
                \node[circle, draw, minimum size=1.5em] (D1) at (6,0) {};
                \node[circle, draw, minimum size=1.5em] (D2) at (6,-2) {1};
                \node[circle, draw, minimum size=1.5em] (E1) at (8,0) {};
                \node[circle, draw, minimum size=1.5em] (E2) at (8,-2) {1};
                \node[circle, draw, minimum size=1.5em] (F1) at (10,0) {};
                \node[circle, draw, minimum size=1.5em] (F2) at (10,-2) {2};
                \node[circle, draw, minimum size=1.5em, fill=black!20] (G1) at (12,0) {};
                \node[circle, draw, minimum size=1.5em] (G2) at (12,-2) {1};
                \node[circle, draw, minimum size=1.5em] (H1) at (14,0) {};
                \node[circle, draw, minimum size=1.5em] (H2) at (14,-2) {2};

                \node[] (A) at (-1,0) {A};
                \node[] (B) at (-1,-2) {B};
                \node[] (R1) at (0,-3) {\Large After Round 3};
            
                \draw[] (A1) -- (A2);
                \draw[] (B1) -- (B2);
                \draw[color=red, very thick] (A1) -- (B2);
                \draw[color=red, very thick] (B1) -- (C2);
                \draw[] (C1) -- (C2);
                \draw[color=red, very thick] (C1) -- (D2);
                \draw[] (D1) -- (D2);
                \draw[color=red, very thick] (D1) -- (E2);
                \draw[] (E1) -- (E2);
                \draw[color=red, very thick] (E1) -- (F2);
                \draw[] (F1) -- (F2);
                \draw[color=red, very thick] (F1) -- (G2);
                \draw[] (G1) -- (G2);
                \draw[] (G1) -- (H2);
                \draw[color=red, very thick] (H1) -- (H2);
            
            \end{tikzpicture}
            
        }

        \vspace{10px}

        \scalebox{.8}{
            \begin{tikzpicture}
    
                \node[circle, draw, minimum size=1.5em] (A1) at (0,0) {};
                \node[circle, draw, minimum size=1.5em, fill=black!20] (A2) at (0,-2) {0};
                \node[circle, draw, minimum size=1.5em] (B1) at (2,0) {};
                \node[circle, draw, minimum size=1.5em] (B2) at (2,-2) {1};
                \node[circle, draw, minimum size=1.5em] (C1) at (4,0) {};
                \node[circle, draw, minimum size=1.5em] (C2) at (4,-2) {2};
                \node[circle, draw, minimum size=1.5em] (D1) at (6,0) {};
                \node[circle, draw, minimum size=1.5em] (D2) at (6,-2) {1};
                \node[circle, draw, minimum size=1.5em] (E1) at (8,0) {};
                \node[circle, draw, minimum size=1.5em] (E2) at (8,-2) {1};
                \node[circle, draw, minimum size=1.5em, fill=black!20] (F1) at (10,0) {};
                \node[circle, draw, minimum size=1.5em] (F2) at (10,-2) {2};
                \node[circle, draw, minimum size=1.5em] (G1) at (12,0) {};
                \node[circle, draw, minimum size=1.5em] (G2) at (12,-2) {2};
                \node[circle, draw, minimum size=1.5em] (H1) at (14,0) {};
                \node[circle, draw, minimum size=1.5em] (H2) at (14,-2) {2};

                \node[] (A) at (-1,0) {A};
                \node[] (B) at (-1,-2) {B};
                \node[] (R1) at (0,-3) {\Large After Round 4};
            
                \draw[] (A1) -- (A2);
                \draw[] (B1) -- (B2);
                \draw[color=red, very thick] (A1) -- (B2);
                \draw[color=red, very thick] (B1) -- (C2);
                \draw[] (C1) -- (C2);
                \draw[color=red, very thick] (C1) -- (D2);
                \draw[] (D1) -- (D2);
                \draw[color=red, very thick] (D1) -- (E2);
                \draw[] (E1) -- (E2);
                \draw[color=red, very thick] (E1) -- (F2);
                \draw[] (F1) -- (F2);
                \draw[] (F1) -- (G2);
                \draw[color=red, very thick] (G1) -- (G2);
                \draw[] (G1) -- (H2);
                \draw[color=red, very thick] (H1) -- (H2);
            
            \end{tikzpicture}
            
        }
        
        \caption{The state of a path on 16 vertices after each setup round. Edges in red are part of the matching. Unfilled vertices are matched, gray vertices are unmatched. The numbers in the $B$-vertices represent their cost. \label{fig:setup}}
    \end{figure}
    After the fourth round, the graph is in the following state. Each $a_k$ for $k \leq 3x - 5$ is matched to $b_{k+1}$, $a_{3x-3}$ is matched to $b_{3x-3}$ and $a_{3x-2}$ is matched to $b_{3x-2}$. $a_{3x-4}$ is unmatched. All vertices in $B_0$ have cost 1, except $b_0$ which has cost 0 and $b_{3x-3}$ which has cost 2. All vertices in $B_1$ have cost 1, except $b_{3x-2}$ which has cost 2. All vertices in $B_2$ have cost 2. 

    During the remainder of the execution of the algorithm, we will use the same tie-break property as described in the second round. That is, the algorithm will take the edge connecting the unmatched $A$-vertex to the $B$-vertex of higher index.

    To lower bound the number of rounds required for a maximum matching, we will show that using the tie-break condition above causes the algorithm to repeatedly re-match the same vertices multiple times. We will call each of these processes a \textit{sweep}. 
    
    The conditions required for a sweep to begin are the following, where $a_k$ is the unmatched vertex.
    
    \begin{enumerate}
        \item[\textbf{C1:}] $b_k$ has cost at least that of $b_{k+1}$.
        \item[\textbf{C2:}] For all $i > k$, the cost of $b_{i+1}$ is either the same as the cost of $b_i$, or one more.
    \end{enumerate}

    We note that both of these conditions hold for the state of the matching after the fourth round. This means that a sweep will begin after the setup rounds. We now describe what occurs during a sweep. An example of a sweep is given in Figure~\ref{fig:sweep}.
    
    The algorithm matches $a_k$ to $b_{k+1}$, the cost of $b_{k+1}$ increases by 1, and $a_{k+1}$ becomes unmatched. By Condition \textbf{C2}, the cost of $b_{k+2}$ was at most the cost of $b_{k+1}$ plus one before this round. Therefore after this round the cost of $b_{k+2}$ is at most the cost of $b_{k+1}$ afterwards. This is Condition \textbf{C1} for unmatched vertex $a_{k+1}$. All the other costs are unchanged, so Condition \textbf{C2} remains true. This means that the sweep will continue by matching $a_{k+1}$ to $b_{k+2}$. This must then continue all the way down the path until the endpoint $a_{3x-2}$ becomes unmatched. We will call this process the \textit{rightward} portion of the sweep.

    Once $a_{3x-2}$ becomes unmatched, the only option for the algorithm in the next round is to immediately rematch it to $b_{3x-2}$, and so $a_{3x-3}$ becomes unmatched again. At the start of the sweep, $b_{3x-2}$ had cost at most that of $b_{3x-3}$ plus 1 by Condition \textbf{C2}. At this point the cost of $b_{3x-3}$ has increased by 1 and the cost of $b_{3x-2}$ has increased by 2. Thus, the cost of $b_{3x-2}$ must now be strictly larger than the cost of $b_{3x-3}$, and so $a_{3x-3}$ must match to $b_{3x-3}$ and increase the cost of $b_{3x-3}$ by 1. We can now apply the same argument again to the state of the matching after this round, so $a_{3x-4}$ matches to $b_{3x-4}$. This process, matching $a_k$ to $b_k$, must therefore continue down the path until $a_k$ becomes unmatched again. We will call this the \textit{leftward} portion of the sweep. The state of the matching is now the same as before the sweep, except with all the costs of $b_i$ for $i > k$ increased by 2. 

        \begin{figure}[h]
        \centering
$ $ 
        \scalebox{.8}{
            \begin{tikzpicture}
    
                \node[circle, draw, minimum size=1.5em] (A1) at (0,0) {};
                \node[circle, draw, minimum size=1.5em, fill=black!20] (A2) at (0,-2) {0};
                \node[circle, draw, minimum size=1.5em] (B1) at (2,0) {};
                \node[circle, draw, minimum size=1.5em] (B2) at (2,-2) {1};
                \node[circle, draw, minimum size=1.5em, fill=black!20] (C1) at (4,0) {};
                \node[circle, draw, minimum size=1.5em] (C2) at (4,-2) {2};
                \node[circle, draw, minimum size=1.5em] (D1) at (6,0) {};
                \node[circle, draw, minimum size=1.5em] (D2) at (6,-2) {2};
                \node[circle, draw, minimum size=1.5em] (E1) at (8,0) {};
                \node[circle, draw, minimum size=1.5em] (E2) at (8,-2) {2};
                \node[circle, draw, minimum size=1.5em] (F1) at (10,0) {};
                \node[circle, draw, minimum size=1.5em] (F2) at (10,-2) {3};
                \node[circle, draw, minimum size=1.5em] (G1) at (12,0) {};
                \node[circle, draw, minimum size=1.5em] (G2) at (12,-2) {4};
                \node[circle, draw, minimum size=1.5em] (H1) at (14,0) {};
                \node[circle, draw, minimum size=1.5em] (H2) at (14,-2) {4};

                \node[] (A) at (-1,0) {A};
                \node[] (B) at (-1,-2) {B};
                \node[] (R1) at (1,-3) {\Large Before the sweep};
            
                \draw[] (A1) -- (A2);
                \draw[] (B1) -- (B2);
                \draw[color=red, very thick] (A1) -- (B2);
                \draw[color=red, very thick] (B1) -- (C2);
                \draw[] (C1) -- (C2);
                \draw[] (C1) -- (D2);
                \draw[color=red, very thick] (D1) -- (D2);
                \draw[] (D1) -- (E2);
                \draw[color=red, very thick] (E1) -- (E2);
                \draw[] (E1) -- (F2);
                \draw[color=red, very thick] (F1) -- (F2);
                \draw[] (F1) -- (G2);
                \draw[color=red, very thick] (G1) -- (G2);
                \draw[] (G1) -- (H2);
                \draw[color=red, very thick] (H1) -- (H2);
            
            \end{tikzpicture}
        }
        
        \vspace{10px}
        
        \scalebox{.8}{
            \begin{tikzpicture}
    
                \node[circle, draw, minimum size=1.5em] (A1) at (0,0) {};
                \node[circle, draw, minimum size=1.5em, fill=black!20] (A2) at (0,-2) {0};
                \node[circle, draw, minimum size=1.5em] (B1) at (2,0) {};
                \node[circle, draw, minimum size=1.5em] (B2) at (2,-2) {1};
                \node[circle, draw, minimum size=1.5em] (C1) at (4,0) {};
                \node[circle, draw, minimum size=1.5em] (C2) at (4,-2) {2};
                \node[circle, draw, minimum size=1.5em] (D1) at (6,0) {};
                \node[circle, draw, minimum size=1.5em] (D2) at (6,-2) {3};
                \node[circle, draw, minimum size=1.5em] (E1) at (8,0) {};
                \node[circle, draw, minimum size=1.5em] (E2) at (8,-2) {3};
                \node[circle, draw, minimum size=1.5em] (F1) at (10,0) {};
                \node[circle, draw, minimum size=1.5em] (F2) at (10,-2) {4};
                \node[circle, draw, minimum size=1.5em] (G1) at (12,0) {};
                \node[circle, draw, minimum size=1.5em] (G2) at (12,-2) {5};
                \node[circle, draw, minimum size=1.5em, fill=black!20] (H1) at (14,0) {};
                \node[circle, draw, minimum size=1.5em] (H2) at (14,-2) {5};

                \node[] (A) at (-1,0) {A};
                \node[] (B) at (-1,-2) {B};
                \node[] (R1) at (2.1,-3) {\Large After the rightward portion};
            
                \draw[] (A1) -- (A2);
                \draw[] (B1) -- (B2);
                \draw[color=red, very thick] (A1) -- (B2);
                \draw[color=red, very thick] (B1) -- (C2);
                \draw[] (C1) -- (C2);
                \draw[color=red, very thick] (C1) -- (D2);
                \draw[] (D1) -- (D2);
                \draw[color=red, very thick] (D1) -- (E2);
                \draw[] (E1) -- (E2);
                \draw[color=red, very thick] (E1) -- (F2);
                \draw[] (F1) -- (F2);
                \draw[color=red, very thick] (F1) -- (G2);
                \draw[] (G1) -- (G2);
                \draw[color=red, very thick] (G1) -- (H2);
                \draw[] (H1) -- (H2);
            
            \end{tikzpicture}
        }

        \vspace{10px}
        
        \scalebox{.8}{
            \begin{tikzpicture}
    
                \node[circle, draw, minimum size=1.5em] (A1) at (0,0) {};
                \node[circle, draw, minimum size=1.5em, fill=black!20] (A2) at (0,-2) {0};
                \node[circle, draw, minimum size=1.5em] (B1) at (2,0) {};
                \node[circle, draw, minimum size=1.5em] (B2) at (2,-2) {1};
                \node[circle, draw, minimum size=1.5em, fill=black!20] (C1) at (4,0) {};
                \node[circle, draw, minimum size=1.5em] (C2) at (4,-2) {2};
                \node[circle, draw, minimum size=1.5em] (D1) at (6,0) {};
                \node[circle, draw, minimum size=1.5em] (D2) at (6,-2) {4};
                \node[circle, draw, minimum size=1.5em] (E1) at (8,0) {};
                \node[circle, draw, minimum size=1.5em] (E2) at (8,-2) {4};
                \node[circle, draw, minimum size=1.5em] (F1) at (10,0) {};
                \node[circle, draw, minimum size=1.5em] (F2) at (10,-2) {5};
                \node[circle, draw, minimum size=1.5em] (G1) at (12,0) {};
                \node[circle, draw, minimum size=1.5em] (G2) at (12,-2) {6};
                \node[circle, draw, minimum size=1.5em] (H1) at (14,0) {};
                \node[circle, draw, minimum size=1.5em] (H2) at (14,-2) {6};

                \node[] (A) at (-1,0) {A};
                \node[] (B) at (-1,-2) {B};
                \node[] (R1) at (2,-3) {\Large After the leftward portion};
            
                \draw[] (A1) -- (A2);
                \draw[] (B1) -- (B2);
                \draw[color=red, very thick] (A1) -- (B2);
                \draw[color=red, very thick] (B1) -- (C2);
                \draw[] (C1) -- (C2);
                \draw[] (C1) -- (D2);
                \draw[color=red, very thick] (D1) -- (D2);
                \draw[] (D1) -- (E2);
                \draw[color=red, very thick] (E1) -- (E2);
                \draw[] (E1) -- (F2);
                \draw[color=red, very thick] (F1) -- (F2);
                \draw[] (F1) -- (G2);
                \draw[color=red, very thick] (G1) -- (G2);
                \draw[] (G1) -- (H2);
                \draw[color=red, very thick] (H1) -- (H2);
            
            \end{tikzpicture}
        }
        
        \caption{An example of a path during stages of a sweep. Edges in red are part of the matching. Unfilled vertices are matched, grey vertices are unmatched. The numbers in the $B$-vertices represent their cost. \label{fig:sweep}}
        \end{figure}
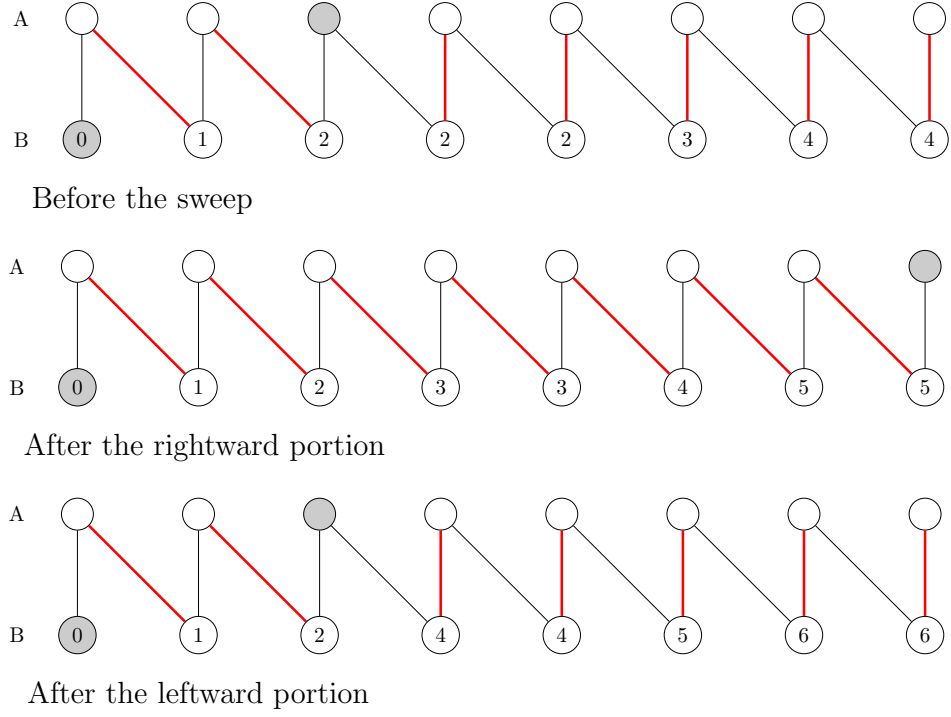

    We now show that sweeps are triggered regularly during the execution of the algorithm. We assume a sweep has just ended, and that before the sweep we had that $c(b_k) = c(b_{k+1}) = 2$. Therefore in the next round, $a_k$ matches to $b_k$ as the cost of $b_{k+1}$ is now 4, compared to the cost of 2 for $b_k$. The algorithm only ever triggers sweeps at $a_k$ after $a_{k+1}$ has just been unmatched, so the vertices with indices less than $k$ have not been altered since the end of the setup rounds. Therefore, in the next round $a_{k-1}$ will be matched to $b_{k-1}$ as it has cost 1 compared to cost 3 for $b_k$. In the subsequent round, $a_{k-2}$ will be matched to $b_{k-2}$ as it has cost 1 while $b_{k-1}$ now has cost 2. However, in the next round $a_{k-3}$ has both neighbours with cost 2. This satisfies Condition \textbf{C1} for a sweep to begin. Condition \textbf{C2} is also satisfied. This is because the $B$-vertices with indices greater than $k$ have had each of their costs increased by 2 since the start of the previous sweep, and this does not affect this condition holding. Furthermore, the items from $b_{k-2}$ to $b_{k+1}$ have costs $2,2,3,4$ which also satisfy this condition. Therefore a new sweep must start at $a_{k-3}$.

    This means that a sweep is triggered once for each vertex of $A_2$. We can use the lengths of these sweeps to lower bound the number of rounds required for a maximum matching.

    \begin{theorem}\label{lem:lower_max}
        The \textsc{ALT}-algorithm requires $\Omega(n^2)$ rounds in order to return a maximum matching.
    \end{theorem}
    \begin{proof}
        Recall that the algorithm performs a sweep once for each vertex of $A_2$. If a sweep starts at $a_k$, it matches each $b_i$ for $i > k$ exactly twice during the sweep. The number of passes required for each of the sweeps is therefore $2((3x-2) - k)$. Summing up the lengths of each of the sweeps, we find a lower bound for the number of passes required.
        \begin{align*}
            2\sum_{i = 1}^{x-1}(3x-2)-(3i-1) = 2\sum_{i=1}^{x-1}3i-1 \in \Omega(x^2) = \Omega(n^2).
        \end{align*}
    \end{proof}

We can now use this result to bound the number of passes required for a $(1-\epsilon)$-approximation.

\begin{corollary}
    The \textsc{ALT}-algorithm requires $\Omega(\frac{1}{\epsilon^2})$ rounds in order to return a $(1-\epsilon)$ approximation to a maximum matching, with or without the freezing mechanism.
\end{corollary}
\begin{proof}
    Choose $x$ to be the integer such that $3x-1<\frac{1}{\epsilon}\leq 3x+2$. The only $(1-\epsilon)$-approximation to a maximum matching on this graph can only be the singular maximum matching, as we have that $(1-{\epsilon})(3x-1) > (1 - \frac{1}{3x-1})(3x-1) = 3x-2$. By Theorem~\ref{lem:lower_max}, the \textsc{ALT}-algorithm requires $\Omega(x^2) = \Omega(\frac{1}{\epsilon^2})$ passes in order to return a maximum matching. 

    Additionally, since the vertex costs are always non-decreasing at each step of the algorithm, each vertex cost is bounded above by the cost of $b_{3x-2}$ at the end of the execution of the algorithm. This vertex is matched twice during the setup phase, and twice during each sweep. Therefore, the cost of each vertex at the end of the execution is at most $2 + 2(x-1) = 2x$. For $x \geq 3$ this is at most $3x - 1 < \frac{1}{\epsilon}$. Hence the freezing mechanism is not triggered, and so the proof also applies when it is included.
\end{proof}

\section{Conclusion}\label{sec:conclusion}
In this paper, we advance the understanding of the \textsc{ALT}-algorithm through three contributions. First, we simplify the algorithm by showing that the freezing mechanism as used in its original version by Assadi et al. is unnecessary. Second, we provide an alternative analysis based on augmenting paths, which explains the algorithm's behaviour more naturally and deals with the removal of the freezing mechanism in a natural way. Third, we construct the first hard instance on which the algorithm indeed requires $\Omega(\frac{1}{\epsilon^2})$ iterations.

We conclude with two open questions. 

First, and most importantly, are there algorithms that solely rely on the computation of maximal matchings that requires fewer than $\Theta(\frac{1}{\epsilon^2})$ iterations/matching computations to find a $(1-\epsilon)$-approximation? 

Second, can we prove lower bounds for such algorithms? For example, can we show that $\Omega(\frac{1}{\epsilon})$ matching computations are necessary to obtain a $(1-\epsilon)$-approximation?

\section*{Acknowledgements}
We thank Xinwei Zhang for pointing out a gap in our analysis that becomes visible when considering  bipartite graphs without perfect matchings. 

\bibliography{bibliography}
\end{document}